\documentclass [12pt]{article} 
\usepackage{amsmath,amsthm}
\usepackage{amssymb,latexsym}
\usepackage{graphicx}
\usepackage{setspace}
\usepackage[mathscr]{eucal}

\numberwithin{equation}{section}

\newcommand{\bee}{\begin{equation}}
\newcommand{\ene}{\end{equation}}
\newcommand{\al}{\alpha}
\newcommand{\pa}{\partial}
\newcommand{\la}{\lambda}
\newcommand{\ep}{\epsilon}

\newtheorem{theorem}{Theorem}

\begin{document} 

\title{Bifurcations and Competing Coherent Structures in the Cubic-Quintic Ginzburg-Landau Equation I: Plane Wave (CW) Solutions}

\vskip0.4in 

\author{Stefan C. Mancas\\
and\\
S. Roy Choudhury \\ 
Department of Mathematics \\ 
Embry Riddle Aeronautical University\\ 
Daytona Beach, FL 32114 \\ 
mancass@erau.edu\\
choudhur@longwood.cs.ucf.edu}

\date{}

\maketitle

\begin{abstract} 

\vspace{.1in}

Singularity Theory is used to comprehensively investigate the bifurcations of the steady-states of the traveling wave ODEs of the cubic-quintic Ginzburg-Landau equation (CGLE). These correspond to plane waves of the PDE. In addition to the most general situation, we also derive the degeneracy conditions on the eight coefficients of the CGLE under which the equation for the steady states assumes each of the possible quartic (the quartic fold and an unnamed form), cubic (the pitchfork and the winged cusp), and quadratic (four possible cases) normal forms for singularities of codimension up to three. Since the actual governing equations are employed, all results are globally valid, and not just of local applicability. In each case, the recognition problem for the unfolded singularity is treated. The transition varieties, i.e. the hysteresis, isola, and double limit curves are presented for each normal form. For both the most general case, as well as for various combinations of coefficients relevant to the particular cases, the bifurcations curves are mapped out in the various regions of parameter space delimited by these varieties. The multiplicities and interactions of the plane wave solutions are then comprehensively deduced from the bifurcation plots in each regime, and include features such as regimes of hysteresis among co-existing states, domains featuring more than one interval of hysteresis, and isola behavior featuring dynamics unrelated to the primary solution branch in limited ranges of parameter space.

\end{abstract}

\doublespace

\section{Introduction}

The cubic complex Ginzburg-Landau equation (CGLE) is the canonical equation governing the weakly nonlinear behavior of dissipative systems in a wide variety of disciplines \cite{1}. In fluid mechanics, it is also often referred to as the Newell-Whitehead equation after the authors who derived it in the context of B\'enard convection \cite{1,2}.

As such, it is also one of the most widely studied nonlinear equations. Many basic properties of the equation and its solutions are reviewed in \cite{3,4}, together with applications to a vast variety of phenomena including nonlinear waves, second-order phase transitions, superconductivity, superfluidity, Bose-Einstein condensation, liquid crystals and string theory. The numerical studies by Brusch et al. \cite{5,6} which primarily consider periodic traveling wave solutions of the cubic CGLE, together with secondary pitchfork bifurcations and period doubling cascades into disordered turbulent regimes, also give comprehensive summaries of other work on this system. Early numerical studies \cite{7} and theoretical investigations \cite{8} of periodic solutions and secondary bifurcations are also of general interest for our work here.

Certain situations or phenomena, such as where the cubic nonlinear term is close to zero, may require the inclusion of higher-order nonlinearities leading to the so-called cubic-quintic CGLE. This has proved to be a rich system with very diverse solution behaviors. In particular, a relatively early and influential review by van Saarloos and Hohenberg \cite{9}, also recently extended to two coupled cubic CGL equations \cite{10,11}, considered phase-plane counting arguments for traveling wave coherent structures, some analytic and perturbative solutions, limited comparisons to numerics, and so-called `linear marginal stability analysis' to select the phase speed of the traveling waves. Among the multitude of other papers, we shall only refer to two sets of studies which will directly pertain to the work in this paper. The first of this series of papers \cite{12}-\cite{16} using dynamical systems techniques to prove that the cubic-quintic CGLE admits periodic and quasi-periodic traveling wave solutions. The second class of papers \cite{17,18}, primarily involving numerical simulations of the full cubic-quintic CGL PDE in the context of Nonlinear Optics, revealed various branches of plane wave solutions which are referred to as continuous wave (CW) solutions in the Optics literature. More importantly, these latter studies also found various spatially confined coherent structures of the PDE, with envelopes which exhibit complicated temporal dynamics. In \cite{17}, these various structures are categorized as plain pulses (periodic), pulsating solitary waves and so on depending on the temporal behavior of the envelopes. In addition, the phase speed of the coherent structures may be zero, constant, or periodic (since it is an eigenvalue, the phase speed may be in principle also be quasiperiodic or chaotic, although no such cases appear to have been reported). Secondary complete period doubling cascades leading as usual to regimes of chaos are also found. This last feature for numerical solutions of the full cubic-quintic PDE is strongly reminiscent of the period doubling cascades found in \cite{5,6} for period solutions of the traveling wave reduced ODEs for the cubic CGLE.

Motivated by the above, we begin a fresh look at the traveling wave solutions of the cubic-quintic CGLE in this paper. Besides attempting to understand the complex numerical coherent structures in \cite{18}, one other goal is to build a bridge between the dynamical systems approach in \cite{12}-\cite{16} and the numerical one in \cite{17,18}. Given the importance of the cubic-quintic CGLE as a canonical pattern-forming system, this is clearly important in and of itself. However, a word of warning is in order here. Some of the features in \cite{18} may well be inherently spatio-temporal in nature, so that a spatial traveling-wave reduction may not be sufficient to completely capture all aspects. Indeed, there is some evidence along these lines \cite{19}.

In this paper, we begin by using Singularity Theory \cite{20} to comprehensively categorize the plane wave (CW) solutions which were partially considered numerically in \cite{17}. In addition, we shall be able to identify co-existing CW solutions in all parameter regimes together with their stability. The resulting dynamic behaviors will include hysteresis among co-existing branches, as well as the existence of isolated solution branches (isolas) separated from the main solution branch. Subsequent papers will consider periodic traveling waves (traveling periodic wavetrains of the PDE), quasiperiodic solutions and homoclinic solutions (corresponding to pulse solutions of the PDE) and their bifurcations.

The remainder of this paper is organized as follows. Section 2 considers two formulations for the traveling-wave reduced ODEs for the cubic-quintic CGLE, as well as CW solutions. Section 3 quickly recapitulates the standard stability analysis for individual CW solutions. In Section 4, which is the heart of the paper, Singularity Theory is employed to comprehensively categorize all possible co-existing and competing plane wave solutions in general parameter regimes, as well as special cases corresponding to all possible quartic and cubic normal forms for singularities of codimension up to three. Section 5 considers the corresponding bifurcation diagrams as well as the resulting dynamical behaviors.

\section{Traveling Wave Reduced ODEs}

\subsection{Reductions}

We shall consider the cubic-quintic CGLE in the form \cite{9}
\bee \label{2.1}
\pa_tA=\epsilon A+(b_1+ic_1)\pa_x^2A-(b_3-ic_3) |A|^2A-(b_5-ic_5) |A|^4 A 
\ene
noting that any three of the coefficients (no two of which are in the same term) may be set to unity by appropriate scalings of time, space and $A$.

For the most part, we shall employ the polar form used in earlier treatments \cite{5,9} of the traveling wave solutions of \eqref{2.1}. This takes the form of the ansatz
\begin{align}\label{2.2}
A(x,t) &= e^{-i\omega t} \hat A (x-vt)\notag\\
&= e^{-i\omega t} a(z)e^{i\phi (z)}
\end{align}
where
\bee\label{2.3}
z\equiv x-vt
\ene
is the traveling wave variable and $\omega$ and $v$ are the frequency and translation speed (and are eigenvalues). Substitution of \eqref{2.2}/\eqref{2.3} in \eqref{2.1} leads, after some simplification, to the three mode dynamical system
\begin{subequations}\label{2.4}
\begin{align}
a_z &= b\label{2.4a}\\
b_z &= a\psi^2 -\frac{(b_1\epsilon +c_1\omega)a+v(b_1b+c_1\psi a) -(b_1b_3-c_1c_3)a^3-(b_1b_5-c_1c_5)a^5}{b_1^2+c_1^2}\label{2.4b}\\
\psi_z &= -\frac{2\psi b}{a}+\frac{-b_1\omega +c_1\epsilon +v\left(\frac{c_1b}{a}-b_1\psi\right)-(b_1c_3+b_3c_1)a^2-(b_1c_5+b_5c_1)a^4}{b_1^2+c_1^2} \label{2.4c}
\end{align}
\end{subequations}
where $\psi\equiv\phi_z$. Note that we have put the equations into a form closer to that in \cite{5}, rather than that in \cite{9}, so that \eqref{2.4} is a generalization of the traveling wave ODEs in \cite{5} to include the quintic terms.

For future reference, we also include the fourth-order ODE system one would obtain from \eqref{2.1} using the rectangular representation
\begin{align}\label{2.5}
A(x,t)&=e^{-i\omega t}\hat A (x-vt)\notag\\
&=e^{-i\omega t} [\al (z)+i\beta (z)]
\end{align}
with $z$ given by \eqref{2.3}. Using \eqref{2.5} in \eqref{2.1}  yields the system:
\begin{subequations}
\begin{align}
-c_1\delta_z+b_1\gamma_z&=\Gamma_1 \label{2.6a}\\
b_1\delta_z+c_1\gamma_z&=\Gamma_2 \label{2.6b}
\end{align}
\end{subequations}
where $\gamma =\alpha'$, $\delta =\beta'$, $^\prime =\mathrm{d}/\mathrm{d}z$, and $\Gamma_1/\Gamma_2$ are given below. This may be written as a first order system
\begin{align}\label{2.7}
\al'&=\gamma\notag\\
\beta'&=\delta\notag\\
(b_1^2+c_1^2)\gamma'&=b_1\Gamma_1+c_1\Gamma_2\notag\\
(b_1^2+c_1^2)\delta'&=b_1\Gamma_2-c_1\Gamma_1
\end{align}
with
\begin{subequations}
\begin{align}
\Gamma_1&=\omega\beta -v\gamma -\epsilon\alpha +(b_3\al +c_3\beta )(\al^2+\beta^2 )+(b_5\al+c_5\beta )(\al^2+\beta^2)^2\label{2.8a}\\
\intertext{and}
\Gamma_2&=-\omega\al -v\delta -\epsilon\beta+(b_3\beta -c_3\al )(\al^2+\beta^2)+(b_5\beta -c_5\al )(\al^2+\beta^2)^2 .\label{2.8b}
\end{align}
\end{subequations}

\subsection{Fixed Points and Plane (Continuous) Wave Solutions}

From \eqref{2.2}, a fixed point $(a_0,0,\psi_0)$ of \eqref{2.4} corresponds to a plane wave solution
\bee\label{2.9}
A(x,t)=a_0e^{i(\psi_0z-\omega t)+i\theta}
\ene
with $\theta$ an arbitrary constant. 

The fixed points of \eqref{2.4a}, \eqref{2.4b} and \eqref{2.4c} may be obtained by setting $b=0$ (from \eqref{2.4a}) in the right hand sides of the last two equations, solving the last one for $\psi$, and substituting this in the second yielding the quartic equation
\bee\label{2.10}
\al_4x^4+\al_3x^3+\al_2x^2+\al_1x+\al_0=0
\ene
with
\begin{subequations}\label{2.11}
\begin{align}
x &= a^2,\label{2.11a}\\
\alpha_4&=\frac{(b_1c_5+b_5c_1)^2}{b_1^2v^2}\label{2.11b}\\
\alpha_3&=\frac{2(b_1c_3+b_3c_1)(b_1c_5+b_5c_1)}{b_1^2v^2}\label{2.11c}\\
\alpha_2&=\frac{b_3^2c_1^2+2b_1b_3c_1c_3-2b_5c_1^2\epsilon}{b_1^2v^2}+\frac{b_5v^2+b_1(c_3^2+2c_5\omega)+2c_1(b_5\omega-c_5\epsilon)}{b_1v} \label{2.11d}\\
\alpha_1&=\frac{b_3}{b_1}+\frac{2(b_1\omega-c_1\epsilon)(b_1c_3+b_3c_1)}{b_1^2v^2}  \label{2.11e}\\
\alpha_0&=\frac{(c_1\epsilon-b_1\omega)^2}{b_1^2v^2}-\frac{\epsilon}{b_1} \;.\label{2.11f}
\end{align}
\end{subequations}
Thus, with $a_0=\sqrt{x}$ for $x$ any of the four roots of \eqref{2.10}, we have a plane wave solution of the form \eqref{2.9}. 

The fixed points of the system \eqref{2.7} are given by $\gamma =\delta =0$ and $\Gamma_1 =\Gamma_2=0$. They may be obtained by eliminating the $\al$ and $\beta$ terms by solving $\Gamma_1=\Gamma_2=0$ simultaneously yielding:
$$
\al^2+\beta^2 =0
$$
or
$$
\alpha^2 +\beta^2 = \frac{b_5\omega+c_5\epsilon}{b_3c_5-b_5c_3} .
$$
Resubstituting these into the $\Gamma_1=\Gamma_2=0$ yields only the trivial fixed point 
\bee\label{2.12}
\al =\beta =0 .
\ene
Thus, the system \eqref{2.7} has no non-trivial plane wave solutions.

In the next section, we begin the consideration of the stability, co-existence and bifurcations of the plane wave states of \eqref{2.1} (the fixed points of \eqref{2.4a}, \eqref{2.4b} and \eqref{2.4c}).

\singlespace
\section{Stability Analysis for Individual Plane Wave \\Solutions}

\doublespace
In this section, we conduct a stability analysis of individual plane wave solutions using regular phase plane techniques. This was already done for the alternative formulation of the traveling wave ODEs given in \cite{9}. We provide a brief derivation for our system \eqref{2.4a}, \eqref{2.4b} and \eqref{2.4c} for completeness and future use. However, a much more complex question is the issue of categorizing and elucidating the possible existence of, and transitions among, multiple plane wave states which may co-exist for the same parameter values in \eqref{2.1} (corresponding to the same operating conditions of the underlying system). Such behavior is well-documented in systems such as the Continuous Stirred Tank Reactor System \cite{20,24}. For a system such as \eqref{2.1} and the associated ODEs \eqref{2.4a}, \eqref{2.4b} and \eqref{2.4c}, the large number of parameters makes a comprehensive parametric study of co-existing states bewilderingly complex, if not actually impracticable. This more complex issue is addressed in the next section.

For each of the four roots $x_i,i=1, \dots ,4$ of \eqref{2.10} corresponding to a fixed point of \eqref{2.4a}, \eqref{2.4b} and \eqref{2.4c} or a plane wave $\sqrt{x_i} \; e^{i(\psi_iz-wt)+i\theta_i}$, the stability may be determined using regular phase-plane analysis. The characteristic polynomial of the Jacobian matrix of a fixed point $x_i=a_i^2$ of \eqref{2.4a}, \eqref{2.4b} and \eqref{2.4c} may be expressed as
\bee \label{3.1}
\la^3+\delta_1\la^2+\delta_2\la +\delta_3=0
\ene
where
\begin{subequations}
\begin{align}
\delta_1 &=\frac{2b_1v}{b_1^2+c_1^2} \label{3.2a}\\
\delta_2 &=\frac{1}{b_1^2+c_1^2 } \big[ 3a^2(c_1c_3-b_1b_3)+5a^4(c_1c_5-b_1b_5)\notag\\
&+(b_1\epsilon +c_1\omega)+3(b_1^2+c_1^2)\psi^2 +v(v-3c_1\psi)\big] \label{3.2b}\\
\delta_3 &= \frac{4a^2\psi[(b_1c_3+b_3c_1)+2a^2(b_1c_5+b_5c_1)]}{b_1^2+c_1^2}\notag\\
&+\frac{1}{(b_1^2+c_1^2)^2}\Big\{b_1c_1\psi v^2-v\Big[a^2\Big(2b_3(b_1^2+c_1^2)+b_1(b_1b_3-c_1c_3)\Big)\notag\\
&+a^4\Big(4b_5(b_1^2+c_1^2)+b_1(b_1b_5-c_1c_5)\Big)-b_1\Big((b_1\epsilon+c_1\omega)-\psi^2(b_1^2+c_1^2)\Big)\Big]\Big\},\label{3.2c}
\end{align}
\end{subequations}
where the fixed point values $(a_i,\psi_i)=(\sqrt{x_i}, \psi_i)$ are to be substituted in terms of the system parameters (\S 2). Note that $\psi_i$ is obtained by setting $a=a_i=\sqrt{x_i},$ and $ b=0$ in the right side of \eqref{2.4c}.

Using the Routh-Hurwitz conditions, the corresponding fixed point is stable for 
\bee\label{3.3}
\delta_1>0,\quad\delta_3>0,\quad\delta_1\delta_2-\delta_3>0.
\ene
Equation \eqref{3.3} is thus the condition for stability of the plane wave corresponding to $x_i$.

On the contrary, one may have the onset of instability of the plane wave solution occurring in one of two ways. In the first, one root of \eqref{3.1} (or one eigenvalue of the Jacobian) becomes non-hyperbolic by going through zero for 
\bee \label{3.4}
\delta_3=0.
\ene
Equation \eqref{3.4} is thus the condition for the onset of `static' instability of the plane wave. Whether this bifurcation is a pitchfork or transcritical one, and its subcritical or supercritical nature, may be readily determined by deriving an appropriate canonical system in the vicinity of \eqref{3.4} using any of a variety of normal form or perturbation methods \cite{21}-\cite{23}.

One may also have the onset of dynamic instability (`flutter' in the language of Applied Mechanics) when a pair of eigenvalues of the Jacobian become purely imaginary. The consequent Hopf bifurcation at
\bee\label{3.5}
\delta_1\delta_2-\delta_3=0
\ene
leads to the onset of periodic solutions of \eqref{2.4a}, \eqref{2.4b} and \eqref{2.4c} (dynamic instability or `flutter'). These periodic solutions for $a(z)$ and $\psi (z)$, which may be stable or unstable depending on the super- or subcritical nature of the bifurcation, correspond via \eqref{2.2} to solutions
\bee\label{3.6}
A(x,t)=a(z)e^{i\left(\int \psi dz-\omega t\right)}
\ene
of the CGLE \eqref{2.1} which are, in general, quasiperiodic wavetrain solutions. This is because the period of $\psi$ and $\omega$ are typically incommensurate. Eq. \eqref{3.6} is periodic if $\omega =0$. We shall consider these wavetrains, including the derivation of normal forms, more general versions of the Hopf bifurcation, and stability, in a companion paper.

Here, we change gears to address the more difficult question of the possible coexistence of, and transitions among, multiple plane wave states for the same parameter sets.

\section{Co-existing and Competing Plane Waves}

As mentioned earlier, for a multiparameter system like \eqref{2.1}, and the associated ODEs \eqref{2.4a}, \eqref{2.4b} and \eqref{2.4c}, a comprehensive parametric study of co-existing states is forbiddingly complex, if not actually impracticable. Theoretical guidance is needed to determine all the multiplicity features in various parameter domains, as well as the stability of, and mutual transitions among, coexisting plane waves in each domain.

In this section, we use Singularity Theory \cite{20} to comprehensively analyze such multiplicity features for \eqref{2.1}/\eqref{2.4a}, \eqref{2.4b}, \eqref{2.4c}. In particular, we shall derive the existence conditions on the eight coefficients of the CGLE under which the steady state equation \eqref{2.10} assumes either a. all possible quartic normal forms (the quartic fold, and an unnamed form), or b. all distinct cubic normal forms (the pitchfork or the winged cusp) for singularities of codimension up to three. In addition, given that the most degenerate singularities or bifurcations tend to be the primary organizing centers for the dynamics, we also consider the even higher codimension singularities leading to various quadratic normal forms. Clearly, the most degenerate singularities for a particular parameter set would `organize' the dynamics in the sense that local behavior in its vicinity predicts actual quasi-global results. In fact, since we employ the actual governing equations, the ensuing results are not just locally valid, as is often the case, but they have global applicability.

First, denoting \eqref{2.10} as
\bee\label{4.1}
g(x,\al_i)=\al_4x^4+\al_3x^3+\al_2x^2+\al_1x+\al_0=0
\ene
where $g$ denotes the `germ' and the $\al_i$ are given in terms of system parameters by \eqref{2.11}, all points of bifurcation (where the Implicit Function Theorem fails) satisfy
\bee\label{4.2}
g_x=0.
\ene
Given a germ  satisfying \eqref{4.1}/\eqref{4.2}, the general Classification Theorem in \cite{20} provides a comprehensive list of all possible distinct normal forms to which it may be reduced for bifurcations of codimension less than or equal to three.

For our $g$, which is already in polynomial form, it is particularly straightforward to reduce it to each of these normal forms in turn and this is what we shall do next. Following this, we shall consider the general form \eqref{4.1} itself. We start first with the possible distinct quartic normal forms viz, the `Quartic Fold' and an unnamed form, and then proceed systematically to lower order normal forms. In the standard manner, the so-called `Recognition Problem' or identification of each normal form yields certain defining conditions and non-degeneracy conditions and we check these first for each form. Each normal form has a well-known `universal unfolding' or canonical form under any possible perturbation \cite{20}. This is so under certain other non-degeneracy conditions (the conditions for the solution of the so-called recognition problem) which we next satisfy. Once the universal unfolding is established, we next need to consider the various parameter regions (for the parameters in the unfolding) where distinct behaviors for the solutions $x$ occur. The boundaries of these regions are the so-called `transition varieties' across which these behaviors change or are non-persistent. We consider these next. The final step involves detailing in each region delimited by two adjacent `transition variety' curves the bifurcation diagram for $x$, i.e., the possible co-existing steady states of \eqref{2.4a}, \eqref{2.4b} and \eqref{2.4c} (or plane waves of \eqref{2.1}) and their stability.

\subsection{The Quartic Fold}

We perform the steps mentioned above for the first quartic normal form, viz. the Quartic Fold
\bee\label{4.3}
h_1(x,\la )=\overline \epsilon x^4 +\delta\la .
\ene
Clearly, our germ (4.1) has this form for
\begin{align}\label{4.4}
\al_4 &= \overline \epsilon,\notag\\
\al_3 &= \al_2 =\al_1 =0\notag\\
\al_0 &= \delta\la .
\end{align}
For the normal form \eqref{4.3}, the universal unfolding is
\bee\label{4.5}
G_1(x,\la )=\overline \epsilon x^4 +\delta\la +\al x+\beta x^2
\ene
with defining conditions
\bee
g_{xx}=g_{xxx}=0,
\ene
non-degeneracy conditions
$$
\overline \epsilon = \text{sgn} \left( \frac{\partial^4 h_1}{\partial x^4}\right), \quad \delta =\text{sgn} \left( \frac{\partial h_1}{\partial\la}\right)
$$
and provided the condition for the solution of the recognition problem
\bee\label{4.7}
\begin{vmatrix}
g_\la & g_{\la x} &g_{\la xx}\\
G_{1\al} &G_{1\al x} &G_{1\al xx}\\
G_{1\beta} &G_{1\beta x} &G_{1\beta xx} \end{vmatrix} \neq 0
\ene
is satisfied. Given \eqref{4.1} and \eqref{4.4}, the conditions \eqref{4.5} are automatically satisfied, while \eqref{4.7} yields the condition
\bee\label{4.8}
\delta\neq 0.
\ene

The transition varieties across which the $(\la ,x)$ bifurcation diagrams change character are:

\noindent
{\bf i. The Bifurcation Variety}

\bee\label{4.9}
\mathscr{B}  = \{ \vec \al \in R^k : (x,\la ) \;\;\text{such that}\;\; G=G_x=G_\la =0 \;\;\text{at}\;\; (x,\la ,\al )\} .
\ene

\noindent
{\bf ii. The Hysteresis Variety}

\bee\label{4.10}
\mathscr{H} =\{ \vec \al \in R^k : (x,\la ) \;\;\text{such that}\;\; G=G_x=G_{xx}=0 \;\;\text{at}\;\; (x,\la ,\al )\}.
\ene
and,

\noindent
{\bf iii. The Double Limit Variety}

\bee\label{4.11}
\mathscr{D} =\{ \vec \al \in R^k:(x_1,x_2,\la ),x_1 \neq x_2 \;\;\text{such that}\;\; G=G_x=0 \;\;\text{at}\;\; (x_i,\la ,\al ),i=1,2\} .
\ene

We compute these here since the derivations are not given in \cite{20}. For $\mathscr{B}$, we need
$$
G_{1x}=4\overline \epsilon x^3 +\al +2\beta x=0
$$
and
$$
G_{1\la}=\delta =0.
$$
However, $\delta\neq 0$ by \eqref{4.8}, and hence the bifurcation set is just the null set
\bee\label{4.12}
\mathscr{B}=\varnothing .
\ene
For $\mathscr{H}$, we need
\begin{align*}
G_{1x} &= 4\overline \epsilon x^3+\al +2\beta x=0\\
G_{1xx} &= 12\overline \epsilon x^2+2\beta =0
\end{align*}
which yield
\bee\label{4.13}
\mathscr{H} =\left\{ \left(\frac{\al}{8\overline \epsilon} \right)^2 =-\left(\frac{\beta}{6\overline \epsilon} \right)^3, \;\; \beta\leq 0\right\} .
\ene
Similarly, using \eqref{4.11}, it is straightforward to derive the double limit set
\bee\label{4.14}
\mathscr{D} =\{\al =0,\beta\leq 0\} .
\ene
In the $(\al ,\beta )$ plane, the $(\la ,x)$ bifurcation diagrams change character across the curves \eqref{4.12}, \eqref{4.13} and \eqref{4.14}, so that there are different multiplicities of steady-states in the regions they de-limit. We shall consider this in detail in the next section.

\subsection{A Second Quartic Normal Form}

Repeating the above steps for the other possible distinct normal form
\bee\label{4.15}
h_2(x,\la )=\overline \ep x^4+\delta\la x,
\ene
our germ \eqref{4.1} takes this form for
\begin{align}\label{4.16}
\al_4&=\overline \ep\notag\\
\al_3&=\al_2=\al_0=0\notag\\
\al_1&=\delta\la .
\end{align}
For the normal form \eqref{4.15}, the universal unfolding is 
\bee\label{4.17}
G_2(x,\la )=\overline \ep x^4+\delta\la x+\al +\beta\la +\gamma x^2
\ene
with defining conditions
\bee\label{4.18}
g_{xx}=g_{xxx}=g_\la =0,
\ene
non-degeneracy conditions which are automatically satisfied, and the solution of the recognition problem yielding the condition
\bee\label{4.19}
\begin{vmatrix}
0 &0 &g_{x\la} &0 &g_{xxxx}\\
0 &g_{\la x} &g_{\la\la} &g_{\la xx} &g_{\la xxx}\\
G_{2\al} &G_{2\al x} &G_{2\al\la} &G_{2\al xx} &G_{2\al xxx}\\
G_{2\beta} &G_{2\beta x} &G_{2\beta\lambda} &G_{2\beta xx} &G_{2\beta xxx}\\
G_{2\gamma} &G_{2\gamma x} &G_{2\gamma\lambda} &G_{2\gamma xx} &G_{2\gamma xxx} \end{vmatrix} \neq 0.
\ene
For \eqref{4.1}/\eqref{4.16}, \eqref{4.18} is satisfied, while \eqref{4.19} yields 
\bee\label{4.20}
\overline \ep \delta\neq 0, \quad\text{or}\quad \al_1\al_4\neq 0 .
\ene

We derive the transition varieties for this case since derivations are not provided in \cite{20}. 

\noindent
\underline{For $\mathscr{B}$},
\begin{align*}
G_{2x}&= 4\overline \ep x^3 +\delta\la +2\gamma x=0\\
G_{2\la}&= \delta x+\beta =0
\end{align*}
which, together with \eqref{4.17}, yield
\bee\label{4.21}
\mathscr{B}: \frac{\overline \ep \beta^4}{\delta^4} +\frac{\gamma\beta^2}{\delta^2} +\al=0 .
\ene
\underline{For $\mathscr{H}$}:
$$
G_{2xx}=0 \Rightarrow\gamma =-6\overline \ep x^2
$$
and
$$
G_{2x} = 0\Rightarrow\delta\la =8\overline \ep x^3 .
$$
Together, these yield
$$
\la^2=-8\gamma^3 /27\delta^2 \overline \ep .
$$
Using these in \eqref{4.17} yields the hysteresis curve:
\bee\label{4.22}
\mathscr{H}:\left(\alpha +\frac{\gamma^2}{12 \overline \epsilon}\right)^2 +\frac{8\gamma^3\beta^2}{27\delta^2\overline \epsilon} =0.
\ene
Similarly, the double limit curve $\mathscr{D}$ is:
\bee\label{4.23}
\mathscr{D}:4\alpha =\gamma^2, \quad \gamma\leq 0 .
\ene

In the next two subsections, we summarize similar results for the two distinct cubic normal forms, but omit the details. Then we briefly mention the four possible quadratic normal forms for even more degenerate cases, before concluding the section with the general, least degenerate case.

\subsection{The Pitchfork}

For our germ \eqref{4.1} to have the cubic normal form for the well-known pitchfork bifurcation
\bee\label{4.24}
h_3(x,\la )=\overline \ep x^3+\delta\la x
\ene
we require
\begin{gather}
\al_4=\al_2=\al_0=0\notag\\
\al_3=\overline \ep , \quad \al_1=\delta\la .
\end{gather}
This will have a universal unfolding \cite{20}
\bee\label{4.26}
G_3=\overline \ep x^3+\delta\lambda x+\al +\beta x^2
\ene
provided
$$
\begin{vmatrix} 0 &0 &h_{3x\la} &h_{3xxx}\\
0 &h_{3\la x} &h_{3\la\la} &h_{3\la xx}\\
G_{3\al} &G_{3\al x} &G_{3\al\la} &G_{3\al xx}\\
G_{3\beta} &G_{3\beta x} &G_{3\beta\la} &G_{3\beta xx} \end{vmatrix} =\delta\neq 0 .
$$
The well-known transition varieties, generalized to our notation, are:
\begin{align}
\mathscr{B} &: \quad \al =0\label{4.27}\\
\mathscr{H} &: \quad \al =\beta^3/27\overline \ep^2 \label{4.28}\\
\mathscr{D} &: \quad \varnothing .
\end{align}

\subsection{The Winged Cusp}

The other distinct cubic normal form
\bee\label{4.30}
h_4(x,\la )=\overline \ep x^3+\delta\la^2
\ene
requires
\begin{align}\label{4.31}
\al_4 &=\al_2=\al_1=0\notag\\
\al_3&=\overline \ep , \;\; \al_0 =\delta\la^2 .
\end{align}
This has a universal unfolding \cite{20}
\bee\label{4.32}
G_4(x,\la )=\overline \ep x^3+\delta\la^2+\al +\beta x+\gamma\la x
\ene
provided
$$
\begin{vmatrix} 0 &0 &h_{4x\la} &h_{4xxx}\\
0 &h_{4\la x} &h_{4\la\la} &h_{4\la xx}\\
G_{4\al} &G_{4\al x} &G_{4\al\la} &G_{4\al xx}\\
G_{4\beta} &G_{4\beta x} &G_{4\beta\la} &G_{4\beta xx} \end{vmatrix} =-12\delta\ep \neq 0.
$$
The transition varieties, for our unfolding $G_4$, are:
\begin{align}
\mathscr{B}  &: \quad \alpha =2x^3-\frac{\gamma^2}{4} x^2, \quad \beta =-3x^2+\gamma^2x/2 \label{4.33}\\
\mathscr{H}  &: \quad \alpha\gamma^2 +\beta^2=0, \quad \alpha \leq 0 \label{4.34}\\
\mathscr{D}  &: \quad \varnothing \label{4.35}.
\end{align}

\subsection{Quadratic Normal Forms}

Since our system of ODEs has many parameters, we may clearly have more degenerate (higher codimension) cases corresponding to any of the distinct quadratic normal forms
\begin{align}
h_5(x,\la ) &= \overline \ep x^2+\delta\la\label{4.36}\\
h_6(x,\la ) &= \overline \ep (x^2-\la^2 )\label{4.37}\\
h_7(x,\la ) &= \overline \ep (x^2+\la^2 )\label{4.38}\\
h_8(x,\la ) &= \overline \ep x^2+\delta\la^3\label{4.39}
\end{align}
or
\bee\label{4.40}
h_9(x,\la ) = \overline \epsilon x^2 +\delta \la^4
\ene
Each of these is obtained by matching our germ \eqref{4.1} to the appropriate form, with the defining and non-degeneracy conditions automatically being satisfied (because \eqref{4.1} is polynomial). Solving the recognition problem \cite{20}, the corresponding unfoldings are respectively
\begin{align}
G_5(x,\la )&=\overline \ep x^2+\delta\la\label{4.41}\\
G_{6,7}(x,\la )&=\overline \ep (x^2+\delta\la^2 +\al )\label{4.42}
\end{align}
(with $\delta <0$ for \eqref{4.37} and $\delta >0$ for \eqref{4.38})
\begin{align}
G_8(x,\la)&=\overline \ep x^2+\delta\la^3 +\al +\beta\la\label{4.43}\\
G_9(x,\la )&=\overline \ep x^2+\delta\la^4+\al +\beta\la +\gamma\la^2\label{4.44}
\end{align}
with determinant conditions \cite{20} for the cases \eqref{4.43} and \eqref{4.44} which may be straightforwardly enforced as in previous cases. The $\mathscr{B}$, $\mathscr{H}$, and $\mathscr{D}$ curves for these cases are straightforward generalizations of those given in \cite{20}, and they may be derived as for the quartic and cubic cases.

\subsection{General Case}

Finally, we include the most general possibility where, for arbitrary parameters in the CGLE \eqref{2.1}, we have the germ \eqref{4.1} with all $\al_i$ non-zero. Treating \eqref{4.1} itself as the unfolding, with $\al_0$ the bifurcation parameter $\la$, the transition varieties in the $(\al_1,\al_2)$ plane are:
\begin{align}
\mathscr{B} \;: \quad &\varnothing\label{4.45}\\
\mathscr{H} \;: \quad &\al_2=-6\al_4x^2-3\al_3x\notag\\
&\al_1=8\al_4x^3+3\al_3x^2\label{4.46}\\
\mathscr{D} \;: \quad &\text{identical to }\mathscr{H} \text{ (see Theorem 1 below)}\label{4.47}
\end{align}

\begin{theorem} The Double Limit Variety for \eqref{4.1} is identical to the Hysteresis Variety of \eqref{4.46}.
\end{theorem}

\begin{proof} Using \eqref{4.1} and \eqref{4.11}, $\mathscr{D}$ is defined by the equations
\begin{subequations}
\begin{align}
G(x_1,\la )&=0\label{4.48a}\\
G(x_2,\la )&=0\label{4.48b}\\
G_x(x_1,\la )&=0\label{4.48c}\\
G_x(x_2,\la )&=0 .\label{4.48d}
\end{align}
\end{subequations}
Canceling the trivial solution $x_1=x_2$, the equations obtained from the difference of \eqref{4.48a} and \eqref{4.48b}, and of \eqref{4.48c} and \eqref{4.48d}, yield respectively
\begin{subequations}
\begin{align}
\alpha_1&=-a(2b-a^2)\alpha_4-b\alpha_3-a\alpha_2\label{4.49a}\\
\alpha_2&=-\frac12 (4b\alpha_4+3a\alpha_3)\label{4.49b}
\end{align}
\end{subequations}
where $a\equiv x_1+x_2^2$, $b\equiv x_1^2+x_1x_2+x_2^2$. Using \eqref{4.49b} in \eqref{4.49a}, these yield
\begin{subequations}\label{4.50}
\begin{align}
\alpha_1&=\left(\frac{3a^2}{2} -b\right)\alpha_3+a^3\alpha_4\label{4.50a}\\
\alpha_2&=-\frac{3a}{2} \alpha_3-2b\alpha_4 .\label{4.50b}
\end{align}
\end{subequations}
The equations \eqref{4.48a} and \eqref{4.48b} may be considered to define the bifurcation parameter $\al_0$ which we do not require here. However, \eqref{4.48c} and \eqref{4.48d} independently define $\al_1$ and thus far only their difference has been used. In order to incorporate $\al_1$, we consider the sum of \eqref{4.48c} and \eqref{4.48d} written in terms of a and b as:
$$
4\al_4a(3b-2a^2)+3\al_3(2b-a^2)+2\al_2a+2\al_1=0.
$$
Using \eqref{4.50} in this equation and simplifying yields
\bee\label{4.51}
b=\frac{3a^2}{4}.
\ene
Using this in \eqref{4.50}  yields the parametric equations for $\mathscr{D}$:
\begin{subequations}
\begin{align}
\alpha_1&=a^2\big(\frac{3\alpha_3}{4}+\alpha_4a\big)\label{4.52a}\\
\alpha_2&=-\frac{3a}{2}\big(\alpha_3+ \alpha_4a\big) .\label{4.52b}
\end{align}
\end{subequations}

The re-parametrization $a=2x$ puts this into exactly the form \eqref{4.46} of the hysteresis variety, thus proving the claim.
\end{proof}

Note that the $\mathscr{H}$ curve is parametrized in terms of $x$ (with $\al_3, \al_4$ being chosen values). Also, given the non-degenerate nature of this general case, it is not surprising that there is only one distinct transition variety.

\section{Bifurcation Diagrams and Effects on the Dynamics}

Having mapped out the $\mathscr{B}$, $\mathscr{H}$, and $\mathscr{D}$ curves for the various possible distinct quartic and cubic normal forms, we now proceed in this section to consider the various bifurcation diagrams in the regions which they define in $(\alpha ,\beta )$ space. These will then give us the multiplicities and stabilities of the various co-exisitng steady states of \eqref{2.4a}, \eqref{2.4b} and \eqref{2.4c} (or plane wave solutions of \eqref{2.1}) in each region. In turn, these also enable us to consider dynamic features of the plane wave solutions. The dynamics will include hysteretic behaviors among co--existing plane waves. We will also find regimes of \underline{isolated} plane wave behavior, both for a plane wave branch which co-exists with other branches but cannot interact with them, as well as those which actually occur only in isolation.

We first list examples of representative sets of parameters for which we may have the various degenerate cases considered in Section 4. 

\begin{itemize}
\item[a.] For the Quartic Fold of Section Section 4.1, typical parameters are:
\begin{itemize}
\item[i.] $b_1=0.0845,$ $b_3=-0.0846,$ $b_5=0.0846,$ $c_1=c_3=-c_5=1$, $\epsilon =0.5$, $v=0.1$, $\omega =0$.
\item[ii.] $b_1=b_5=0.01696$, $b_3=-0.0206$, $c_1=1$, $c_3=1.25$, $c_5=-1,$ $\epsilon =0.5$, $v=0.1$, $\omega =0$.
\end{itemize}
\item[b.] For the quartic normal form of Section 4.2:
\begin{itemize}
\item[i.] $b_1=2.035$, $b_3=29.274$, $b_5=9.8496$, $c_1=-0.1$, $c_3=-1$, $c_5=0.08$, $\epsilon =0.3$, $v=0.3$, $\omega =0.1$.
\end{itemize}
\item[c] For the Pitchfork case of Section 4.3:
\begin{itemize}
\item[i.] $b_1=0.0904$, $b_3=0.0679$, $b_5=0.1811$, $c_1=-0.4$, $c_3=0.35$, $c_5=0.8$, $\epsilon =0.2$, $v=0.01$, $\omega =-0.9$.
\item[ii.] $b_1=0.0904$, $b_3=0.0823$, $b_5=-0.1808$, $c_1=-0.4$, $c_3=0.35$, $c_5=-0.8$, $\epsilon =0.2$, $v=0.01$, $\omega =-0.9.$
\end{itemize}
\item[d.] For the Winged Cusp of Section 4.4:
\begin{itemize}
\item[i.] $b_1=0.000923$, $b_3=+.00005548$, $b_5=0.0013$, $c_1=0.5$, $c_3=-0.03$, $c_5=-0.7$, $\epsilon =0.01$, $v=0.1$, $\omega =0.15$.
\end{itemize}
\end{itemize}

For the winged cusp unfolding \eqref{4.32} in the particular form
\begin{equation*}
G_1(x,\lambda )=x^3+\lambda^2+\alpha +\beta x+\gamma\lambda x=0,
\end{equation*}
the transition varieties \eqref{4.33} and \eqref{4.34} are shown in the $(\alpha ,\beta )$ plane in Figure \ref{Figure 7}a,b,c for  $\gamma < 0$, $\gamma =0$, and $\gamma >0$, respectively. They divide the $(\alpha ,\beta )$ space into seven distinct regions. As mentioned earlier, the $(\lambda ,x)$ bifurcation diagrams are isomorphous or `persistent' or of similar form within each region, and they change form across the transition varieties (or `nonpersistence' curves) as one crosses into an adjacent region. The representative bifurcation diagrams in each of the seven regions are shown in Figure \ref {Figure 8}, and they give us a comprehensive picture of the co--existing plane wave solutions of \eqref{2.1} and their stability (given by the eigenvalues of the Jacobian, or here just the sign of $G_x$) in each region. Hence, as we shall consider next, one also has a clear picture of the ensuing dynamics from the plane wave interactions.

\vspace*{.3in}

\begin{figure}
\begin{center} 
\includegraphics[width=350pt]{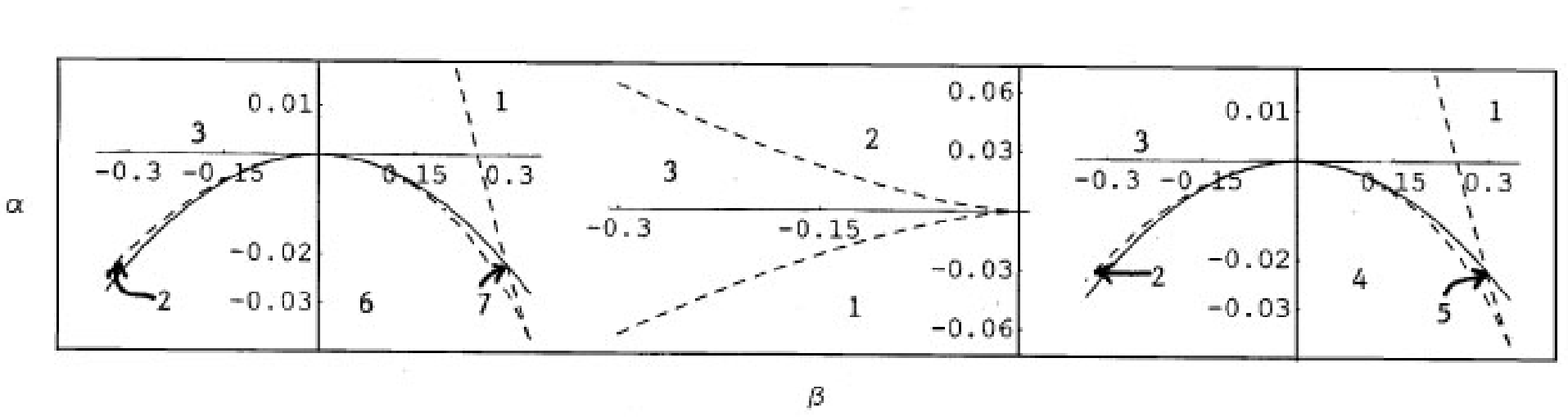}
\end{center}
\caption{Transition varieties for the winged cusp \eqref{4.32} with $\overline \epsilon =1=\delta$ for the cases $\gamma <0$, $\gamma =0$, and $\gamma >0$, respectively. $\mathscr H$ is in solid lines, and $\mathscr B$ is dashed.} \label{Figure 7}
\end{figure}

\begin{figure}
\begin{center} 
\includegraphics[width=350pt]{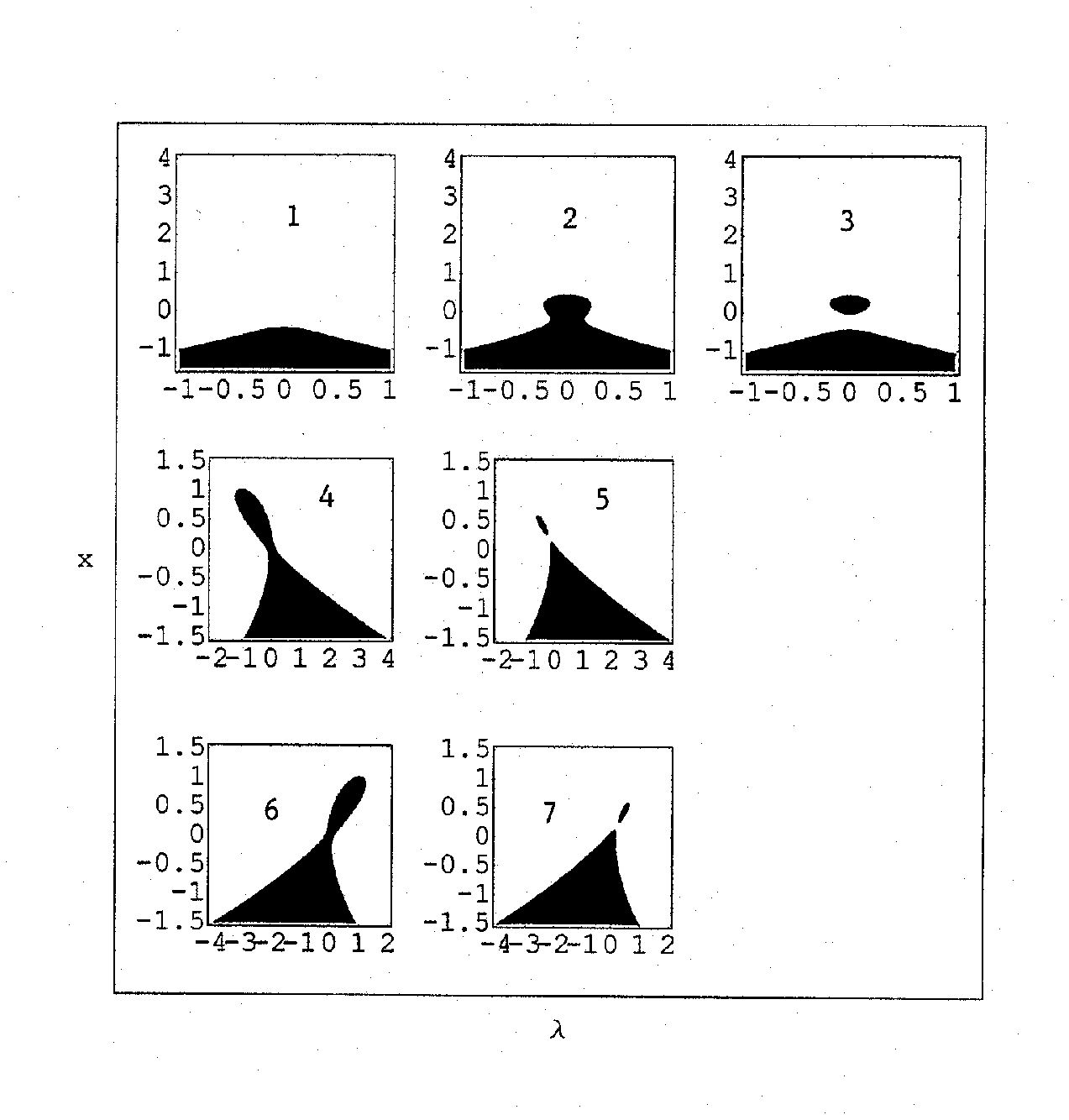}
\end{center}
\caption{The $(\lambda ,x)$ bifurcation diagrams in the regions \texttt{1}-\texttt{7} of Figure \ref{Figure 7}, respectively.}
\label{Figure 8}
\end{figure}
First, note the mushroom shaped bifurcation diagram in Figure \ref{Figure 8}b for region \texttt{2} of Figure \ref{Figure 7}. Clearly, there are two distinct ranges of $\lambda$ (at the two ends of the mushroom) where three plane waves co-exist (with the central one being unstable). Thus the dynamics exhibits hysteresis. For instance, if $\lambda$ is decreased from large values, one stays on the lower branch until point A before jumping to the upper branch. If $\lambda$ is then increased, one stays on the upper branch until B and then jumps back down to the lower one. Similar hysteresis occurs in regions \texttt{4} through \texttt{7} of Figure \ref{Figure 7} as seen in the corresponding bifurcation diagrams of Figure \ref{Figure 8}d--g. In each case, hysteresis occurs between the upper and lower fixed points in the range of $\lambda$ with three co--existing solutions (the central one is always unstable).

Another type of behavior is the isola, i.e., an isolated branch of solutions unconnected to the primary solution (the one at $\lambda\to\pm\infty$). Such isola type behavior is seen in Figure \ref{Figure 8}c,e,g corresponding to regions \texttt{3}, \texttt{5}, and \texttt{7} of Figure \ref{Figure 7} . In each case, the isola co-exists with the primary solution branch and is the chosen branch or not according to the initial conditions. However, once chosen, the dynamics is on the isola while $\lambda$ is in the domain of its existence once we leave this domain, the solution cannot jump to the primary branch and just disappears.

Next, we consider the normal form \eqref{4.15} in Section 4.2. Considering the unfolding \eqref{4.17} in the particular form
$$
G_2(x,\lambda )=x^4-\lambda x+\alpha +\beta\lambda+\gamma x^2=0,
$$
the transition varieties \eqref{4.21}, \eqref{4.22}, and \eqref{4.23} are shown in Figure \ref{Figure 9}a,b,c for the cases $\gamma >0$, $\gamma =0$, and $\gamma <0$ respectively. Note in particular, a significant correction to \cite{3} in the $\mathscr{H}$ curve of Figure \ref{Figure 9}c. The $\mathscr{H}$ curve \eqref{4.22} represents a pair of straight lines in the $(\alpha ,\beta )$ plane, rather than the incorrect form
$$
\alpha+\frac{\gamma^2}{12\overline \epsilon} +\frac{8\gamma^3\beta^2}{27\delta^2\overline \epsilon}=0
$$
in \cite{20}. In Figure \ref{Figure 9}c, one consequence is two new regions or domains \texttt{13} and \texttt{14} of the $(\alpha ,\beta )$ space. Also, the bifurcation plots in the domains \texttt{3}, \texttt{4}, \texttt{5} and \texttt{8} are significantly modified from those given in \cite{20} for the corresponding regions.

\begin{figure}
\begin{center} 
\includegraphics[width=350pt]{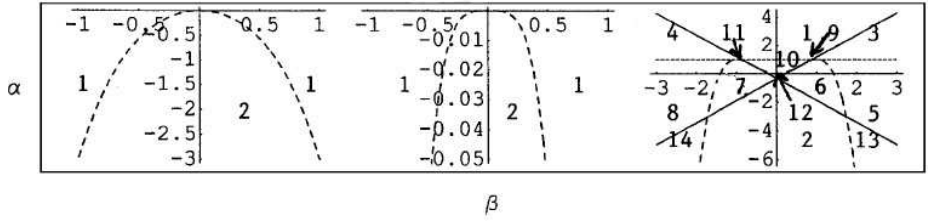}
\end{center}
\caption{The transition varieties for the quartic normal form \eqref{4.17} with  $\delta =-1$ for the cases $\gamma >0$, $\gamma =0$, and $\gamma <0$, respectively. $\mathscr H$ is in solid lines, $\mathscr B$ is dashed and the double limit curve $\mathscr D$ is in fine dashing. The regions \texttt{1}-\texttt{14} which they delimit are shown.}
\label{Figure 9}
\end{figure}

The bifurcations plots in the fourteen regions in Figure \ref{Figure 9}c (Figure \ref{Figure 9}a,b feature only some of the regions) are shown in Figures \ref{Figure 10} and \ref{Figure 11}. Note that there are no regions of isola behavior. In regions \texttt{3}, \texttt{4}, \texttt{5}, and \texttt{8}, there is only one branch of solutions, rather than two as shown in Figure \ref{Figure 10} (case 10) of \cite{20}. Of these, the segments BC and DE are unstable in cases \texttt{3} and \texttt{5}, so that the hysteretic behavior of the solutions will consist of transitions from the stable plane waves on branch AB to those on branch CD as $\lambda$ is increased past point B, and a reverse transition when it is decreased through C. Similarly, in regions \texttt{4} and \texttt{8} where only segment BC is unstable, hysteresis occurs with a transition from the plane wave on branch DE to branch AB if $\lambda$ is decreased through D, a transition from branch CD to branch AB when $\lambda$ is decreased through D, and a transition from CD to either AB or DE (depending on system bias, noise et cetera) as $\lambda$ is increased through C. Analogous hysteresis behavior is clearly possible in regions \texttt{7} and \texttt{11}, while regions \texttt{9}, \texttt{10}, and \texttt{12} feature hysteresis between co--existing stable plane wave solutions on \underline{distinct} solution branches. In the two new regions \texttt{13} and \texttt{14} of Figure \ref {Figure 9}c (which were missing in \cite{20}), the bifurcation plots in Figures \ref{Figure 11}m and \ref{Figure 11}n show only two co--existing plane wave solutions in each $\lambda$ range, unlike the adjacent regions \texttt{5} and \texttt{8} of Figure \ref{Figure 9}c where the bifurcation plots Figures \ref{Figure 10}e and \ref{Figure 10}h have $\lambda$ ranges with four coeval solutions.

\begin{figure} 
\begin{center} 
\includegraphics[width=350pt]{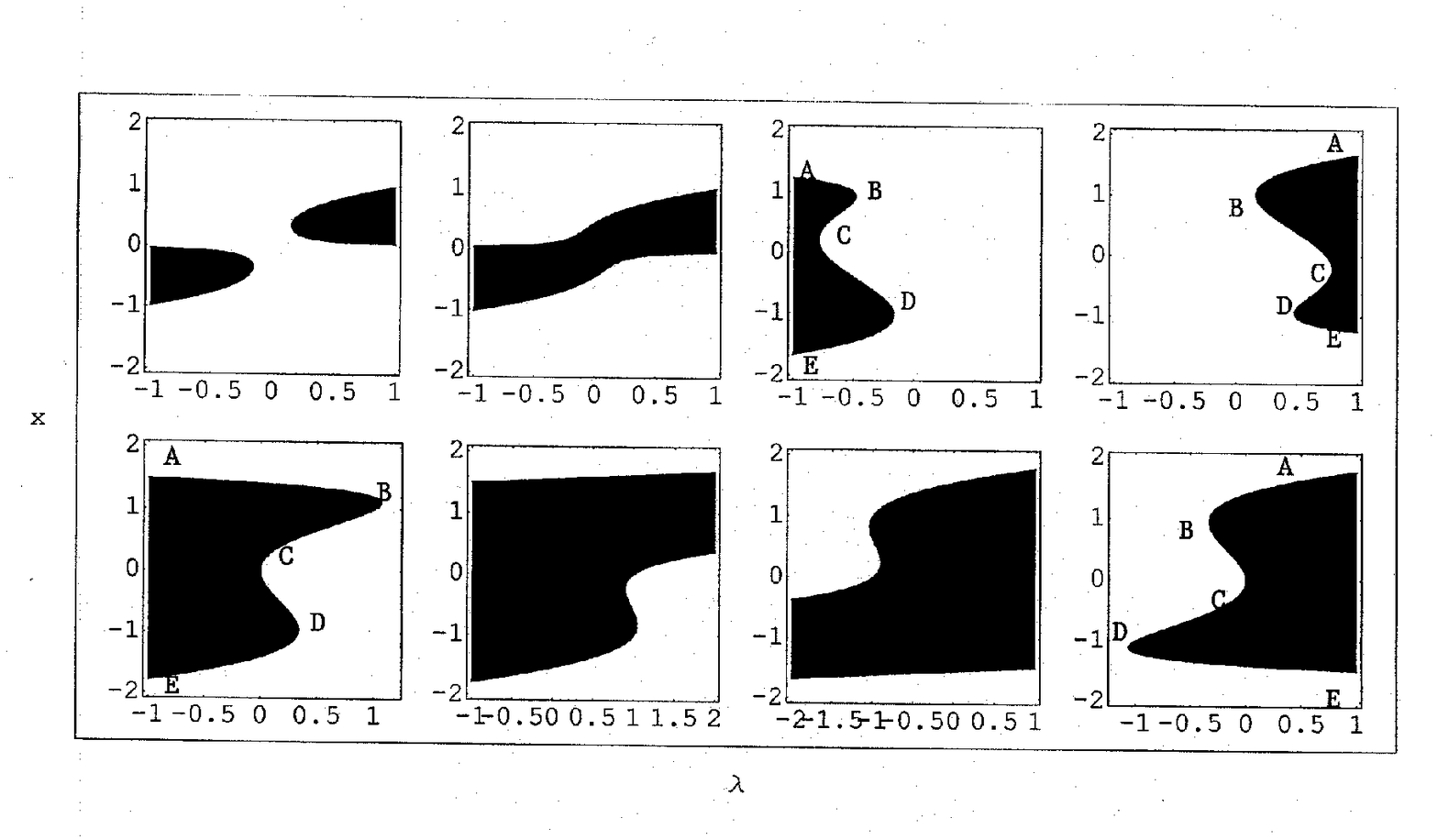}
\end{center}
\caption{Bifurcation diagrams in the regions \texttt{1}-\texttt{8} of Figure \ref{Figure 9}c.}
\label{Figure 10}
\end{figure}

\begin{figure}
\begin{center} 
\includegraphics[width=350pt]{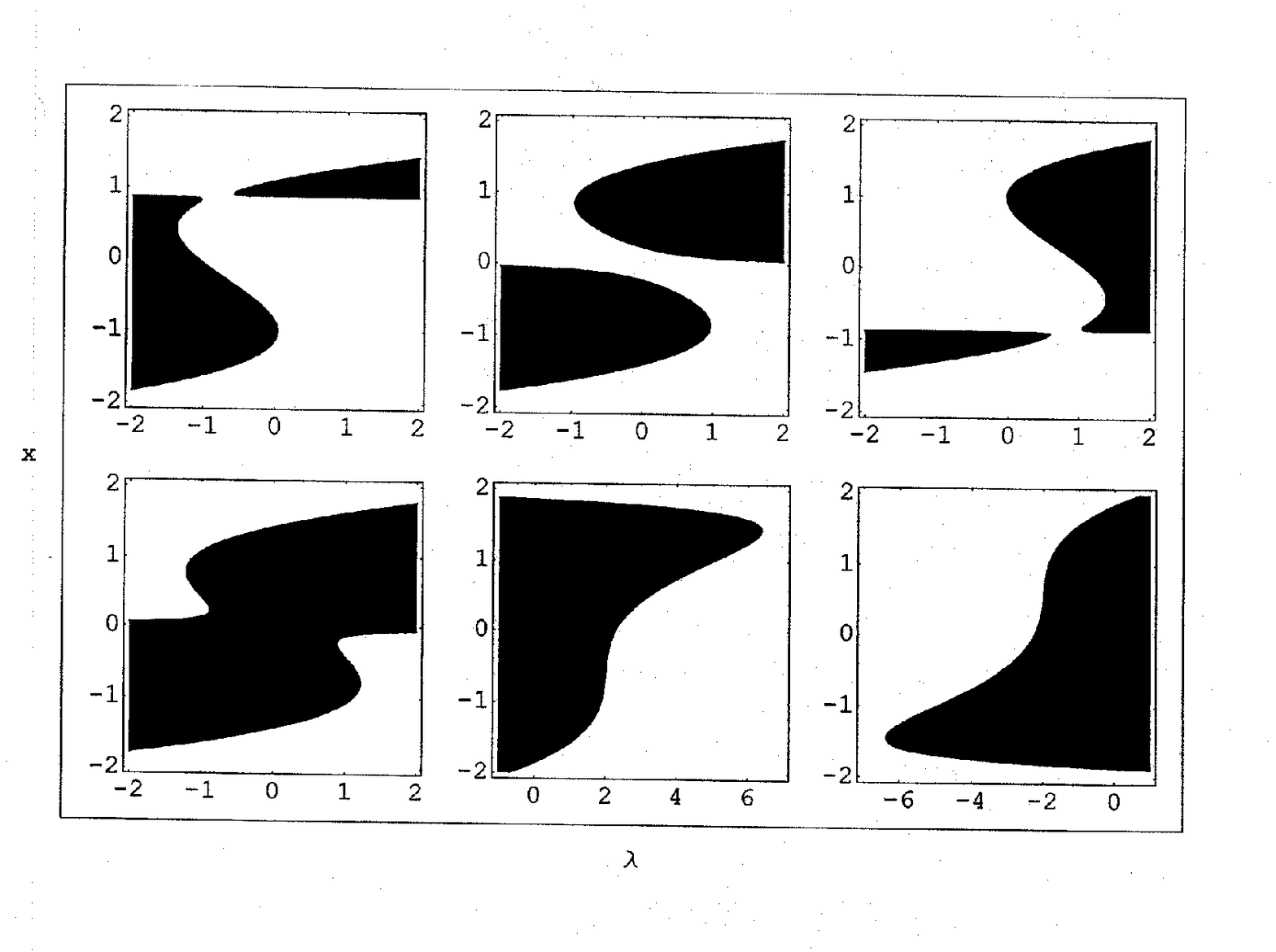}
\end{center}
\caption{Bifurcation diagrams in the regions \texttt{9}-\texttt{14} of Figure \ref{Figure 9}c.}
\label{Figure 11}
\end{figure}

\vspace*{.2in}

For the very degenerate cases discussed in Section 4.5 and corresponding to quadratic normal forms, the corresponding transition varieties as well as the bifurcation plots and resulting dynamics in the regions of $(\alpha ,\beta )$ which they delimit may be deduced from the relevant cases in Figures 4.1-4.3 of \cite{20}. In particular, cases \eqref{4.42}, \eqref{4.43} and \eqref{4.44} show isola, hysteresis, and double isola behaviors respectively. 

Let us consider the general case in Section 4.6 next. Since the results are entirely new, we need to consider the issue of the stability of the $(\lambda ,x)\equiv (\alpha_0, x)$ bifurcation diagrams in various regions of the $(\alpha_1, \alpha_2 )$ plane. Using \eqref{4.1} and the Chain Rule,
\begin{equation*}
\frac{\mathrm{d}G}{\mathrm{d}\lambda} \equiv \frac{\mathrm{d}G}{\mathrm{d}\alpha_0} =\frac{\pa G}{\pa x}\frac{\mathrm{d}x}{\mathrm{d}\lambda} +\frac{\pa G}{\pa\lambda} =0,
\end{equation*}
so that
\begin{equation}\label{5.1}
\frac{\pa G}{\pa x}=-\frac{1}{(\mathrm{d}x/\mathrm{d}\lambda )} .
\end{equation}
Thus, the Jacobian and its eigenvalue $G_x$ (these are identical for a one-dimensional system such as \eqref{4.1}) are negative, and the corresponding fixed-point branch of the $(\lambda ,x)$ plane is stable, for segments of the bifurcation plot where
\begin{equation}\label{5.2}
\frac{\mathrm{d}x}{\mathrm{d}\lambda} >0, \quad\text{stable.}
\end{equation}
Conversely, segments with $\frac{\mathrm{d}x}{\mathrm{d}\lambda}<0$ are unstable.

Finally, let us consider the dynamics and interactions of plane waves for the most general case of Section 4.6. The coincident transition varieties $\mathscr{H}$ and $D$ in \eqref{4.46}/\eqref{4.47} are shown in Figures \ref{Figure 12}a-\ref{Figure 12}h for various combinations of $(\alpha_3,\alpha_4)$ values. As is readily apparent, the configurations in Figures \ref{Figure 12}a-c are the independent ones corresponding to centered and off-centered cusps and a parabolic variety curve respectively -- the other cases are simple reflections of these. For Figure \ref{Figure 12}a  with $(\alpha_3,\alpha_4) =(0,1)$, the transition variety divides the $(\alpha_1,\alpha_2)$ space into two distinct regions \texttt{1} and \texttt{2}. The bifurcation plots in the two regions are shown in Figures \ref{Figure 13}a and \ref{Figure 13}b. As per \eqref{5.1}, the segment(s) with $\frac{\mathrm{d}x}{\mathrm{d}\lambda}>0$ are stable, so that there is a unique stable plane wave for Figure \ref{Figure 13}a in region \texttt{1} of Figure \ref{Figure 12}a. By contrast, there are co--existing stable plane wave states in regions BC and DE of Figure \ref{Figure 13}b (for region \texttt{2} of Figure \ref{Figure 12}a). Thus, hysteretic dynamics occurs with a transition from BC to DE as $\lambda\equiv\alpha_0$ is decreased through C, and a reverse transition as $\lambda$ is increased on DE through D. For Figure \ref{Figure 12}b with $(\alpha_3,\alpha_4)=(1,1),$ the bifurcation plots in regions \texttt{1} and \texttt{2} are shown in Figures \ref{Figure 14}a and \ref{Figure 14}b respectively. Once gain, per \eqref{5.1}, the segments of these plots with positive slope correspond to stable plane waves. Thus, only the segment corresponding to DE in Figure \ref{Figure 14}a is a unique stable plane wave solution in region \texttt{1} of Figure \ref{Figure 12}b. For region \texttt{2} of Figure \ref{Figure 12}b, Figure \ref{Figure 14}b shows hysteresis between the stable plane wave branches BC and DE. For regions 1 and 2 of Figure \ref{Figure 12}c corresponding to $(\alpha_3,\alpha_4)=(1,0),$ the bifurcation plots in regions \texttt{1} and \texttt{2} are shown in Figures \ref{Figure 15}a and \ref{Figure 15}b. For the former, as per \eqref{5.1}, no stable plane waves exist. For the latter, there is a unique stable plane wave solution in the range of $\lambda (\alpha_0)$ corresponding to segment BC.

\begin{figure}
\begin{center} 
\includegraphics[width=350pt]{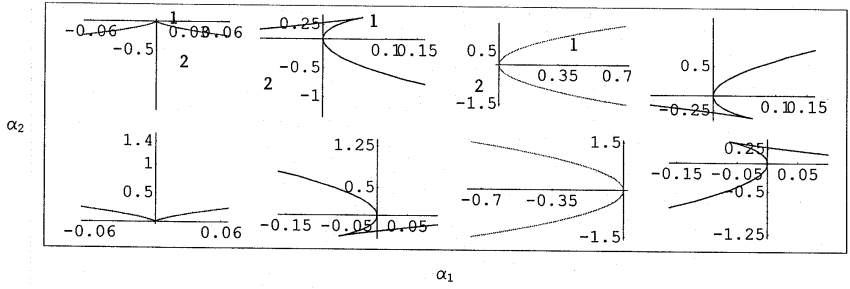}
\end{center}
\caption{Transition varieties for the general case \eqref{4.1} treated in Section 4.6. There is no $\mathscr B$ curve, and $\mathscr H$ and $\mathscr D$ are coincident. The regions they delimit are shown. The figures correspond respectively to $(\alpha_3, \alpha_4)$ values $(0,1)$, $(1,1)$, $(1,0)$, $(-1,1)$, $(-1,0)$, $(-1,-1)$, $(0,-1)$, and $(1,-1)$.}
\label{Figure 12}
\end{figure}

\begin{figure}
\begin{center} 
\includegraphics[width=350pt]{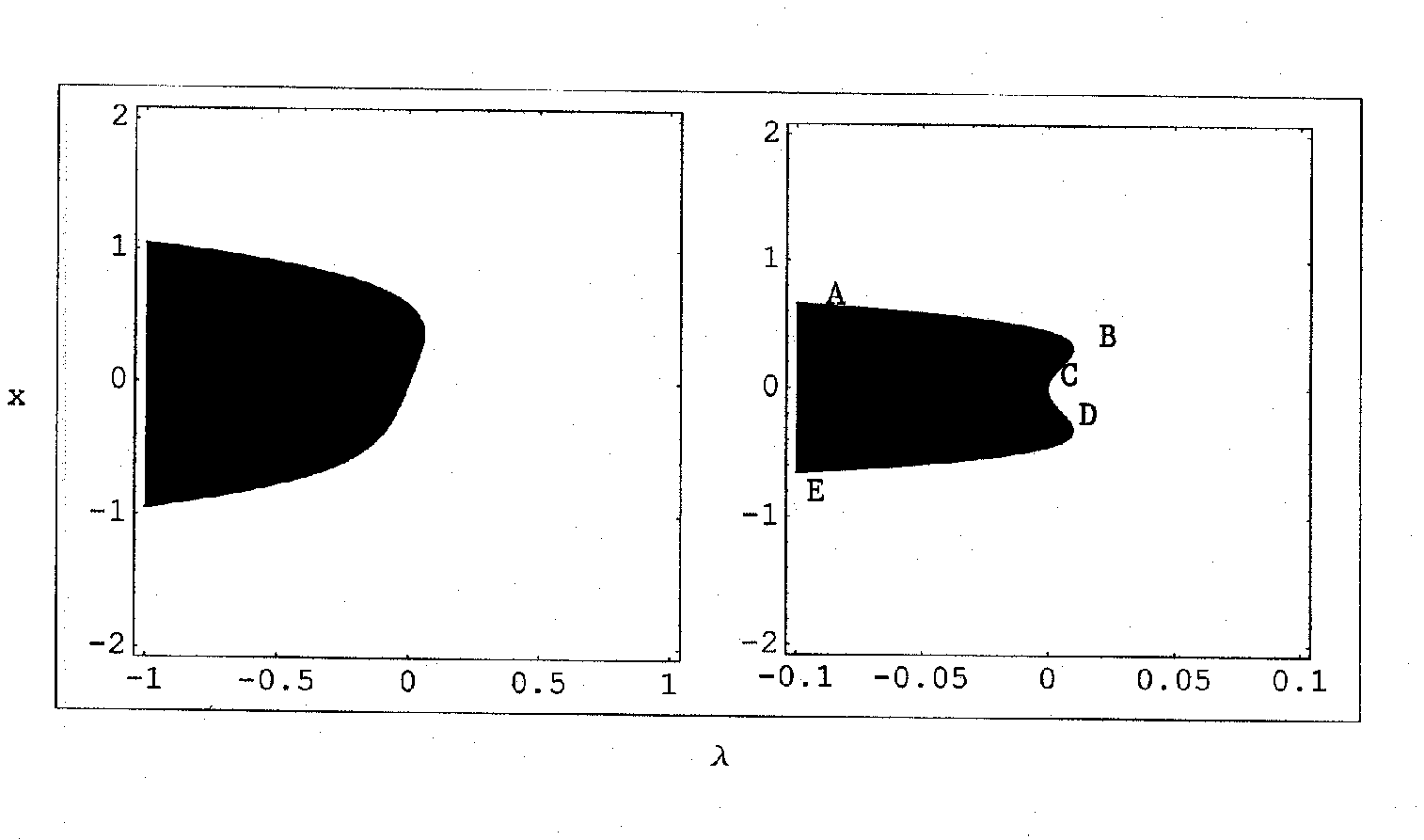}
\end{center} 
\caption{Bifurcation diagrams in regions \texttt{1} and \texttt{2} of Figure \ref{Figure 12}a, respectively.}
\label{Figure 13}
\end{figure}

Finally, for the sake of completeness, we mention an alternative interpretation of the general case using Catastrophe Theory \cite{25} (see \cite{20} for a discussion of the connection between this and the Singularity Theory approach). Treating \eqref{4.1} in a manner analogous to the Cusp Catastrophe,
\begin{align}
G_x&= 4\alpha_4x^3+3\alpha_3x^2+2\alpha_2x+\alpha_1\notag\\
&\equiv 4\alpha_4(x^3+\Gamma_2x^2+\Gamma_1x+\Gamma_0) =0\label{5.3}
\end{align}
with
\begin{align}
\Gamma_2&= \frac{3\alpha_3}{4\alpha_4},\notag\\
\Gamma_1&=\frac{\alpha_2}{2\alpha_4},\notag\\
\Gamma_0&= \frac{\alpha_1}{\alpha_4} . \label{5.4}
\end{align}
Defining
\begin{align}
q&=\frac13 \Gamma_1-\frac19 \Gamma_2^2\notag\\
r &= \frac16 (\Gamma_1\Gamma_2-3\Gamma_0)-\frac{\Gamma_2^3}{27}, \label{5.5}
\end{align}
the transition cusp curve between domains with one and three real solutions is given by
\begin{equation}\label{5.6}
q^3+r^2=0.
\end{equation}
For $(\alpha_3,\alpha_4)=(1,1)$ corresponding to Figures \ref{Figure 12}b and \ref{Figure 14}, the catastrophe surface \eqref{5.3} showing regions of one/three real solutions in the $(\alpha_1,\alpha_2)$ plane shown in Figure \ref{Figure 15}a. Figure \ref{Figure 15}b shows the corresponding cusp surface \eqref{5.6} in $(\alpha_1,\alpha_2)$ space. As mentioned, \cite{20} discusses the relationship between these plots and the Singularity Theory plots given in Figures \ref{Figure 12}b and \ref{Figure 14} for this case.

In concluding, we have comprehensively analyzed the co--existing plane wave solutions in various parameter regimes for the CGLE \eqref{2.1}. This includes transitions among co--existing states involving up to two domains with hysteresis, isolated parameter regimes with isola behavior, and the resulting dynamics. We should also stress that, since our governing equation \eqref{4.1} is of polynomial form, all the results in Sections 4 and 5 are globally (and not just locally) valid in their respective regimes, as of course are the results for the general case.

\begin{figure}
\begin{center} 
\includegraphics[width=350pt]{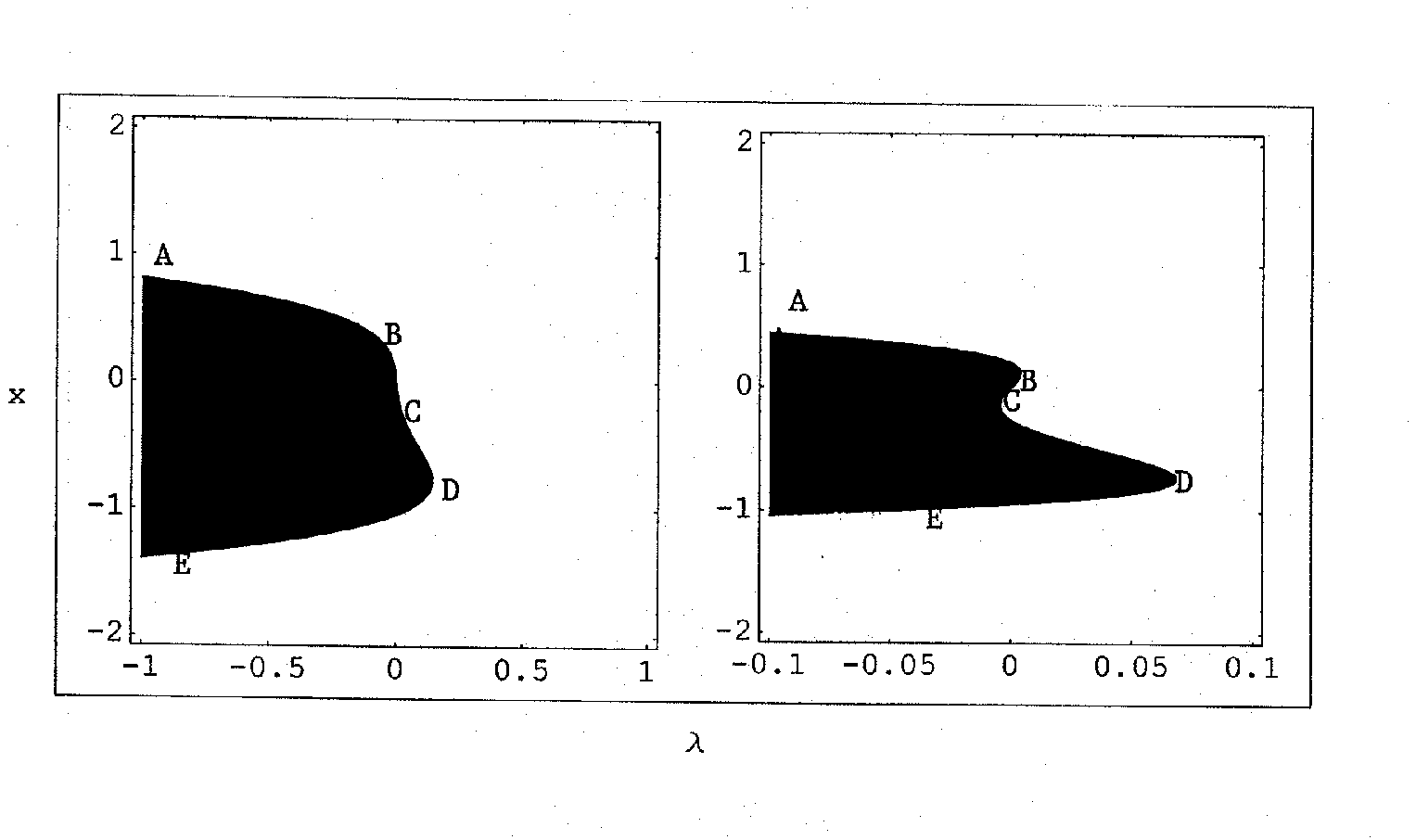}
\end{center}
\caption{Bifurcation diagrams in regions \texttt{1} and \texttt{2} of Figure \ref{Figure 12}b, respectively.}
\label{Figure 14}
\end{figure}

\begin{figure} 
\begin{center} 
\includegraphics[width=350pt]{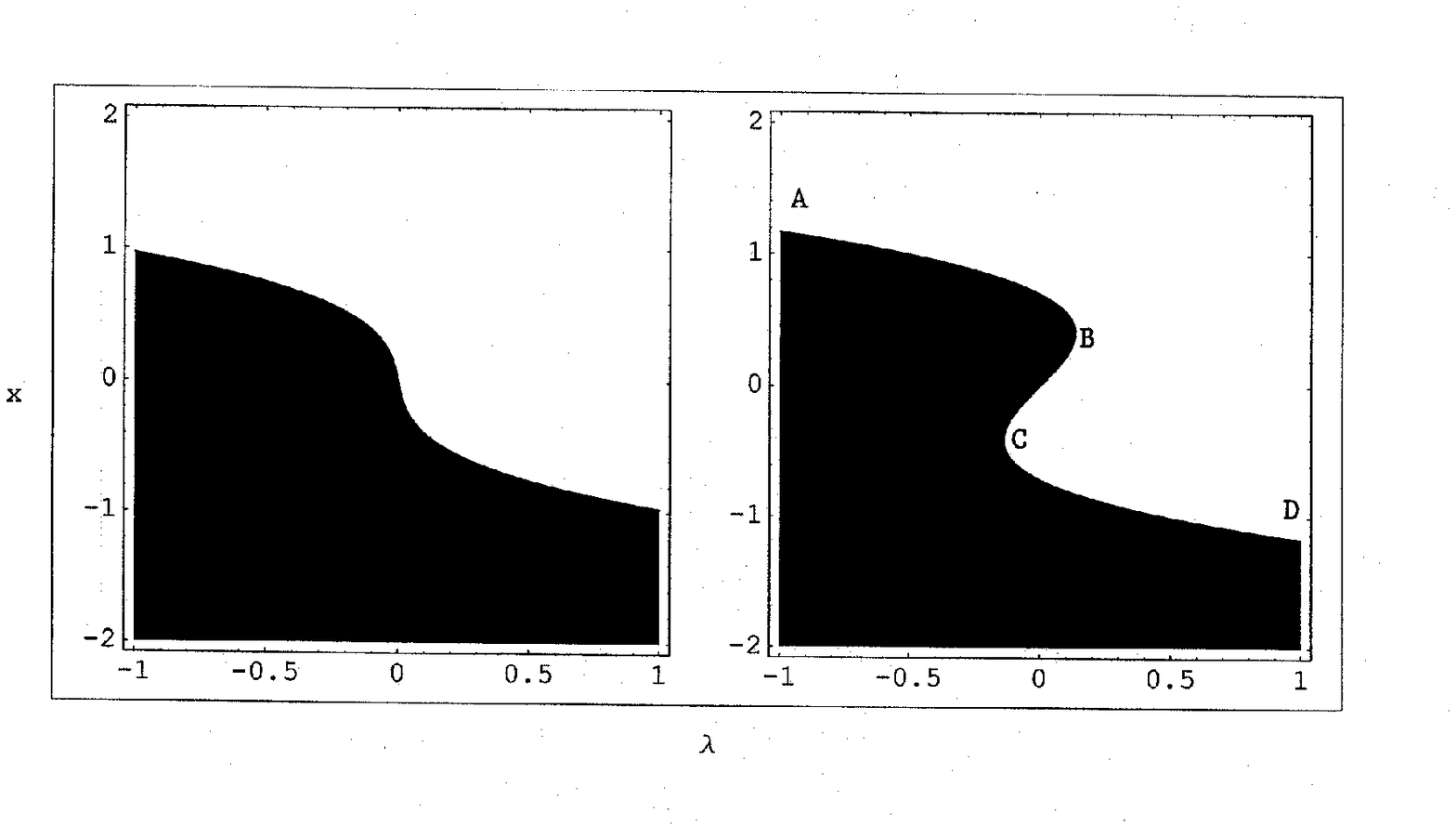}
\end{center}
\caption{Bifurcation diagrams in regions \texttt{1} and \texttt{2} of Figure \ref{Figure 12}c, respectively.} 
\label{Figure 15}
\end{figure}

\singlespace

\end{document}